\documentclass[conference]{IEEEtran}
\IEEEoverridecommandlockouts
\usepackage{cite}
\usepackage{algorithm} 
\usepackage{amsmath,amssymb,amsfonts}
\usepackage{algorithmic}
\usepackage{graphicx}
\usepackage{stfloats}
\usepackage{subfigure}
\usepackage[caption=false,font=normalsize,labelfont=sf,textfont=sf]{subfig}
\usepackage{textcomp}
\usepackage{array}
\usepackage{xcolor}
\usepackage{bm}
\usepackage{graphicx}
\usepackage{setspace}
\usepackage{subfigure}
\usepackage{verbatim}
\usepackage{multicol}
\usepackage{amsmath}
\usepackage{mathrsfs}
\usepackage{amssymb, amsthm}
\usepackage{bm}
\usepackage{cite}
\usepackage{enumerate}
\usepackage{diagbox}
\usepackage{xcolor}
\usepackage{setspace}
\usepackage{multirow}

\newcommand{\cM}{\mathcal{M}}
\newcommand{\cK}{\mathcal{K}}

\newcommand{\cN}{\mathcal{N}}

\newcommand{\cS}{\mathcal{S}}

\newcommand{\cH}{\mathcal{H}}

\newcommand{\bb}{\boldsymbol b}

\newcommand{\bd}{\boldsymbol d}

\newcommand{\texit}[1]{\textit{#1}}

\newtheorem{myth}{Theorem}
\newtheorem{myprop}[myth]{\bf Proposition}

\newtheorem{myle}[myth]{\bf Lemma}

\def\BibTeX{{\rm B\kern-.05em{\sc i\kern-.025em b}\kern-.08em
    T\kern-.1667em\lower.7ex\hbox{E}\kern-.125emX}}

\begin{document}

\title{Enhancing Mobile Crowdsensing Efficiency: A Coverage-aware Resource Allocation Approach}
\author{Yaru Fu\IEEEauthorrefmark{1}, Yue Zhang\IEEEauthorrefmark{2}, Zheng Shi\IEEEauthorrefmark{3}, Yongna Guo\IEEEauthorrefmark{4}, and Yalin Liu\IEEEauthorrefmark{1} \\
\IEEEauthorrefmark{1}School of Science and Technology, Hong Kong Metropolitan University, Hong Kong SAR\\
\IEEEauthorrefmark{2}Department of Electronic and Information Engineering, Shantou University, China\\
\IEEEauthorrefmark{3}School of Intelligent Systems Science and Engineering, Jinan University, China\\
\IEEEauthorrefmark{4}School of Electrical Engineering and Computer Science, KTH Royal Institute of Technology, Stockholm, Sweden\\
\vspace{-0.5cm}
}
\maketitle

\begin{abstract}
In this study, we investigate the resource management challenges in next-generation mobile crowdsensing networks with the goal of minimizing task completion latency while ensuring coverage performance, i.e., an essential metric to ensure comprehensive data collection across the monitored area, yet it has been commonly overlooked in existing studies. To this end, we formulate a weighted latency and coverage gap minimization problem via jointly optimizing user selection, subchannel allocation, and sensing task allocation. The formulated minimization problem is a non-convex mixed-integer programming issue. To facilitate the analysis, we decompose the original optimization problem into two subproblems. One focuses on optimizing sensing task and subband allocation under fixed sensing user selection, which is optimally solved by the Hungarian algorithm via problem reformulation. Building upon these findings, we introduce a time-efficient two-sided swapping method to refine the scheduled user set and enhance system performance. Extensive numerical results demonstrate the effectiveness of our proposed approach compared to various benchmark strategies.
\end{abstract}

\smallskip


\section{Introduction}
As intelligent applications such as virtual reality (VR), augmented reality (AR), extended reality (XR), and beyond gain popularity, the demand for data collection and analysis is increasing significantly. These innovative technologies require vast amounts of real-time data sensing and analysis to function effectively, posing significant challenges for traditional internet-of-things-driven static sensing systems \cite{wenshuai,yfu2023}. 
Mobile crowdsensing (MCS) has emerged as a promising paradigm that leverages the sensing capabilities of mobile devices carried by end-users. 
This paradigm enables large-scale, distributed data collection in real-time, significantly enhancing the spatial and temporal resolution of sensed data and providing a cost-effective solution for monitoring diverse environments \cite{MCS2}.

Prior studies on MCS have largely focused on incentivizing user participation through game theory, auction theory, and machine learning techniques \cite{incentive_review}. 
They proposed a divide-then-conquer approach to tackle the non-convex problem. 
In \cite{us_2}, the utilization of edge networks was explored to support MCS systems in task allocation, taking into account the constraints of limited user energy budgets. 
In a separate study \cite{Li}, authors addressed the sensing reward maximization problem by jointly optimizing task allocation, user selection, and the management of limited wireless network resources such as energy and transmit power. 
The study considered both static and dynamic subband allocation strategies. 
In \cite{wcnc_wenshuai}, unmanned aerial vehicles (UAVs) were leveraged to improve the efficiency of MCS systems. The study focused on optimizing resource allocation for sensing, communication, and computing tasks, alongside UAV trajectory planning to maximize the total computed data volume. 
To tackle this long-term optimization challenge, the authors introduced a novel method based on deep reinforcement learning.

The existing resource management efforts have primarily aimed at maximizing sensing rewards or minimizing energy consumption \cite{WCNC2024,us_2,Li,wcnc_wenshuai}. However, these approaches often neglect the spatial coverage of the scheduled users, leading to incomplete data collection and pottial blind spots in monitored areas. Addressing this gap is essential to ensure that sensing tasks are executed comprehensively across the network.
To this end, we propose a novel resource allocation framework that jointly optimizes user selection, subchannel allocation, and sensing task distribution to minimize overall system latency while maximizing subarea coverage. By formulating this problem as a weighted minimization of latency and coverage gaps, we derive an efficient solution that enhances the performance of MCS networks under practical constraints.

The key contributions of this paper can be summarized as:
First, we formulate a novel framework that minimizes task latency and coverage gaps in mobile crowdsensing networks by jointly optimizing user selection, subband allocation, and task distribution.
Second, the problem is decomposed into two subproblems: (1) optimal task and subband allocation solved by the Hungarian algorithm, and (2) user scheduling solved by a two-sided swapping method.
Last, numerical results show that the proposed approach outperforms existing benchmarks, achieving lower latency and improved coverage across diverse network conditions.

\section{System Model}
In this section, we first introduce the system model for our considered mobile crowdsensing networks. Then, we elaborate on the detailed explanations of the sensing model, transmission model, and coverage performance metric.
\subsection{Network Model}
In our considered mobile crowdsensing network (MCN), a single base station (BS) exists in the center of its coverage area, serving $K$ users. The coverage area is segmented into $M$ subareas, with the set of indices represented as $\cM=\{1,2,\ldots,M\}$. We define $\Phi_k \in \cM$ as the subarea index to which the user $k$ is assigned. The users in this context are equipped for data sensing and can be taken as providers of sensing services. Let $\cK=\{1,2,\ldots,K\}$ represent the index set of all users, and $\cN=\{1,2,\ldots,N\}$ denote the set of indices for all $N$ orthogonal subbands. Adhering to common practices, we assume $K>N$. In other words, when a sensing task is published from the BS with a data size indicated by $\bar{d}$, only a subset of users can be scheduled to execute the sensing task collaboratively. The scheduling of user $k$ is represented by the binary indicator $s_k \in \{0,1\}$. Specifically, $s_k=1$ signifies that user $k$ has been scheduled for the sensing task, and $s_k=0$ otherwise. Each scheduled user is allocated a subband to transmit the sensing data to the BS. For $k\in\cK$ and $n\in\cN$, we define $b_{k,n}\in\{0,1\}$ as the subband allocation indicator. More precisely, $b_{k,n}=1$ if and only if subband $n$ is assigned to user $k$ for sensing data transmission. In addition, we assume that each subband can be utilized by at most one user. Therefore, the following constraint should be satisfied:
\begin{equation}
	\sum_{k\in\cK}b_{k,n} \leq 1,~n\in\cN.
\end{equation}
With the definitions provided earlier, we can establish the relationship between the schedule indicator and the subband allocation indicator as follows:
\begin{equation} \label{ab}
	s_k=\sum_{n\in\cN}b_{k,n}.
\end{equation}
This is indeed accurate, given that only the users scheduled for the task can be allocated subbands for data transmission. Leveraging this relationship, the variable $s_k$ can be substituted with $b_{k,n}$. Further elaboration on this will be provided in the upcoming problem formulation.

\subsection{Data Sensing and Transmission Model}
In the subsequent sections, we discuss the data sensing and transmission model utilized in our investigation. When the BS designates user $k$ to undertake the sensing task, the BS communicates to user $k$ the details of its sensing data bits, denoted as $d_k$. Let $v_k$ represent the sensing rate of user $k$. The overall duration needed for user $k$ to complete the data sensing, denoted by $\text{T}_{k,s}$, is given below:
\begin{equation}\label{se}
	\text{T}_{k,s}=\frac{d_k}{v_k}.
\end{equation} 
Subsequently, we will expound on the transmission delay of user $k$. Let $P_k$ be the transmit power of user $k$, and let $g_{k,n}$ be the channel information of user $k$ on subband $n$. The transmission data rate of user $k$, referred to as $R_k$, is given by the following equation:
\begin{equation}
	R_k=\sum_{n\in\cN}b_{k,n}B_n\log_2(1+\frac{P_kg_{k,n}}{N_0B_n}),
\end{equation} 
where $B_n$ is the bandwidth of subband $n$. Define $\text{T}_{k,f}$ as the transmission delay of user $k$, where $k\in\cK$. We have
\begin{equation} \label{tr}
	\text{T}_{k,f}=\frac{d_k}{R_k}.
\end{equation}
Based on the analysis above, we can calculate the total latency of user $k$, which includes both sensing and transmission delays. This total latency is denoted by $\text{T}_k$, and is given by the following formula:
\begin{equation} \label{T}
	\text{T}_k=\text{T}_{k,s}+\text{T}_{k,f},
\end{equation}
wherein $\text{T}_{k,s}$ and $\text{T}_{k,f}$ are given in \eqref{se} and \eqref{tr}, respectively.
Upon receiving feedback from all scheduled users, the BS will aggregate the results and perform subsequent instructions. Given the high data analysis proficiency of the BS, the aggregation delay is deemed insignificant. Consequently, the total system delay, denoted as $\text{T}_{\text{over}}$, is computed as follows:
\begin{equation}\label{T_o}
	\text{T}_{\text{over}}=\max_{k\in\cK}\text{T}_k,
\end{equation}
in which $\text{T}_k$ is defined in \eqref{T}.
\subsection{Coverage Performance Metric}
In this subsection, we analyze the coverage metric when scheduling users to do task sensing. Considering users' coverage when scheduling tasks in mobile crowdsensing systems is crucial to ensure comprehensive data collection across the monitored area, enhancing data accuracy and completeness. Proper coverage allocation also improves system reliability by mitigating potential data gaps and ensuring effective utilization of available resources. In other words, when selecting users, we try to improve the subarea numbers of all scheduled users. With foregoing definitions, a coverage metric, referred to as $\bar\Phi$, is defined as follows:
\begin{equation}
\bar\Phi=\text{Len}\{s_k\Phi_k|k\in\cK\},
\end{equation}
where $\text{Len}$ is an operation to find the different items in a set. Based on \eqref{ab}, we can rewrite $\bar\Phi$ as follows:
\begin{equation}
\bar\Phi=\text{Len}\{\Phi_k\sum_{n\in\cN}b_{k,n}|k\in\cK\}.
\end{equation}

\section{Problem Formulation}
In this paper, we aim to minimize the overall delay of the MCN while taking into account the coverage of the scheduled users and limited resources. For notation simplicity, we denote $\bb=(b_{k,n})_{k\in\cK,n\in\cN}$ as the subband allocation strategy. Besides, let $\bd=(d_k)_{k\in\cK}$ be the sensing task allocation for the system.  Moreover, we define $\Phi=M-\bar\Phi$ as the coverage gap. With the above-mentioned definitions, the joint optimization problem can be mathematically formulated as follows:   
\vspace{-0.1cm}
\begin{align} \label{obj}
	& \underset{\bb,\bd}{\text{minimize}} ~w\mathrm{Norm}(\text{T}_{\text{over}},\eta)+(1-w)\Phi \\
	\text{s.t.}\nonumber~
	&\text{C1}:~\sum_{k\in\cK}d_k\geq \bar d,~\nonumber\\
	&\text{C2}:~(1-\sum_{n\in\cN}b_{k,n})d_k=0,~k\in\cK,\nonumber\\
	&\text{C3}:~0 \leq d_k\leq \bar d, ~k\in\cK,\nonumber\\
	&\text{C4}:~\sum_{k\in\cK}b_{k,n} \leq 1,~n\in\cN,\nonumber\\
	&\text{C5}:~\sum_{n\in\cN}b_{k,n} \leq 1,~k\in\cK,\nonumber\\
	&\text{C6}:~b_{k,n}\in\{0,1\},~k\in\cK,~n\in \cN,\nonumber
\end{align}
where $\text{T}_{\text{over}}$ in the objective function is introduced in \eqref{T_o}. 
$\mathrm{Norm}(x,\eta)=\frac{2}{1+\exp(-\frac{x}{2\eta})}-1$ is a function normalizing the total delay into $[0,1]$ based on the scaling factor $\eta$  \cite{DeepHuang}.
$w$ is the weight to strike a balance between the significance of latency and coverage. C1 is implemented to ensure the fulfillment of the sensing task initiated by the BS. C2 stipulates that a user will solely be assigned subbands if it has been scheduled for the task. C3 ensures that the sensing task bits allocated to each user remain non-negative and are bounded above by $\bar d$. Furthermore, C4 and C5 constrain each subband to be assigned to at most one user and each user to be allocated at most one subband, respectively.
C6 emphasizes the binary nature of $b_{k,n}$. 
The formulated problem represents a non-convex mixed-integer programming problem due to the non-convexity of C2 concerning $\bb$, and is challenging to be solved. In the next section, we propose an optimal resource allocation approach for problem \eqref{obj} by leveraging its structural characteristics under a specified user scheduling policy. Specifically, with fixed scheduled users, the original minimization problem can be reformulated as a maximum-weighted matching problem within a bipartite graph. The optimal solution to this problem can be derived using the Hungarian algorithm. Subsequently, the primary emphasis is placed on devising the user scheduling strategy, which is addressed through a time-efficient two-sided swapping methodology.

\section{Joint Optimization Algorithm Design} \label{alg}
In this section, we present our joint optimization solution for \eqref{obj}. Initially, we introduce certain definitions that will be pivotal in the ensuing analysis. Define $\cS \subset \cK$ as the set of indexes for the scheduled users. In addition, let $\alpha_k=\frac{1}{v_k}+\frac{1}{R_k}$ for $k\in\cS$. With the subband allocation scheme $\bb$ given, $R_k$ remains constant, rendering $\alpha_k$ achievable. Similarly, $\Phi$ is predetermined based on a specified user scheduling policy. Consequently, the original problem \eqref{obj} simplifies to the subsequent minimization problem:
\begin{align} \label{obj1}
	& \underset{\bd}{\text{min}}~\underset{k\in\cS}{\text{max}}~\alpha_kd_k \\
	\text{s.t.}\nonumber~
	&\text{C1}:~\sum_{k\in\cS}d_k\geq \bar d,~\nonumber\\
	&\text{C3}:~0 \leq d_k\leq \bar d, ~k\in\cS.\nonumber
\end{align}
Let $d^*_k$ represent the optimal sensing data allocation for user $k$, as detailed in the following lemma.
\begin{myle}
	The optimal solution to problem \eqref{obj1} is expressed as $d^*_k=\frac{\bar d}{\alpha_k \left(\sum_{k\in\cS}\frac{1}{\alpha_k}\right)}$, for all $k\in\cS$.
\end{myle}
\begin{proof}
	It can be shown that the optimal solution to problem \eqref{obj1} is achieved when the inequality in C1 is satisfied with equality and the condition $\alpha_id_i=\alpha_jd_j$, $\forall i,j\in\cS$, and $i\neq j$ holds. Let $\alpha_id_i=\beta$. Therefore, we can write:
	\begin{equation}
		\sum_{k\in\cS}\frac{\beta}{\alpha_k}=\bar d,
	\end{equation}
	which implies that $\beta=\frac{\bar d}{\sum_{k\in\cS}\frac{1}{\alpha_k}}$. Substituting this back into $\alpha_id_i=\beta$ completes the proof of the lemma.
\end{proof}

Next, we demonstrate how the original problem, under a given set $\cS$, can be reformulated as a maximum weighted matching problem on a bipartite graph. For each subband $n\in\cN$, let $k_n \in \cK$ represent the user assigned to subband $n$, i.e., $k_n=k$ if $b_{k,n}=1$. Additionally, define $\alpha_{k_n}=\frac{1}{v_{k_n}}+\frac{1}{R_{k_n}}$. Using these definitions, we can present the following theorem:
\begin{myth} \label{myth}
	For a given set $\cS$, the original problem \eqref{obj} can be reformulated as an equivalent optimization problem with the objective:
	\begin{equation}
	\min_{k_1,k_2,\ldots,k_N} \frac{\bar d}{\sum_{n\in\cN}\frac{1}{\alpha_{k_n}}},
	\end{equation}
	subject to the constraints that $k_n \in \cS$ for all $n\in\cN$, and $k_i \neq k_j$ for any $i,j\in\cN$ where $i\neq j$.
\end{myth}

\begin{proof}
	Based on prior analysis, the optimal value of problem \eqref{obj1} can be expressed as $\max_{k\in\cS}\alpha_kd^*_k$. Substituting the result from Lemma 1, $d^*_k=\frac{\bar d}{\alpha_k \left(\sum_{k\in\cS}\frac{1}{\alpha_k}\right)}$, and we find the optimal value to be $\frac{\bar d}{\sum_{k\in\cS}\frac{1}{\alpha_k}}$. 
    Problem \eqref{obj1} assumes a fixed subband allocation. To solve the original problem \eqref{obj}, the pairing between each subband $n$ and user $k_n$ must be optimized further. This completes the proof.
\end{proof}

Theorem \ref{myth} establishes that the original minimization problem can be reformulated as an equivalent optimization problem with a modified objective and constraints. This reformulation simplifies the optimization process. The following proposition provides an efficient approach to solving the reformulated problem:
\begin{myprop}
	By interpreting $\frac{1}{\alpha_{k_n}}$ as the weight between subband $n$ and user $k_n$, the reformulated problem in Theorem \ref{myth} can be treated as a maximum weighted matching problem on a bipartite graph. This formulation inherently satisfies the constraints C4–C6 from the original problem \eqref{obj}. The optimal solution can be efficiently obtained using the Hungarian algorithm, which has a cubic time complexity \cite{com_opt}. The procedure is outlined in Algorithm 1.
\end{myprop}

\begin{algorithm}[t]
\small
\caption{Subband-user pairing optimization}
\begin{algorithmic}[1]
\REQUIRE $\{\frac{1}{\alpha_{k_n}}\}_{k\in\cS,n\in\cN}$
\ENSURE Optimal subband-user pairing \( P \)
\STATE Construct matrix \( \mathbf{U} \) of size \( N \times N \), with elements $U_{n,k} = -\frac{1}{\alpha_{k_n}}$.
\FOR{each row \( n \in \mathcal{N} \)}
    \STATE Subtract the smallest value in row \( n \) from each element.
\ENDFOR
\FOR{each column \( k \in \mathcal{S} \)}
    \STATE Subtract the smallest value in column \( k \) from each element.
\ENDFOR
\STATE Draw lines through the minimum number of rows and columns to cover all zero entries in \( \mathbf{U} \). Let \( h \) denote the number of lines.
\WHILE{\( h < N \)}
    \STATE Identify the smallest uncovered value in \( \mathbf{U} \).
    \STATE Subtract this value from all uncovered elements.
    \STATE Add this value to entries covered by two lines.
    \STATE Update \( h \).
\ENDWHILE
\STATE Construct a bipartite graph \( G \) with vertices representing subbands \( \mathcal{N} \) and users \( \mathcal{S} \). Connect vertices based on zero entries in \( \mathbf{U} \).
\STATE Apply the augmenting path algorithm to find the maximum matching \( P \) of \( G \).
\RETURN Optimal pairing $P$ 
\end{algorithmic}
\end{algorithm}
Next, our attention will shift towards determining $\cS$. We intend to address the user scheduling subproblem using the two-sided swapping method. Building on the preceding analysis, given $\cS$, the optimal solution to the original problem (\ref{obj}) can be derived via Algorithm 1. For analytical ease, we define $\bar\cS = \cK \setminus \cS$. Considering $s' \in {\cS}^{\dag}$ and $s \in \cS$, we interchange the positions of $s'$ and $s$, resulting in an updated user scheduling set, denoted as ${\cS}^{\dag} = \cS\cup\{s'\}\setminus\{s\}$. If ${\cS}^{\dag}$ leads to an improved objective value for problem (\ref{obj}), we assign $\cS = {\cS}^{\dag}$. If not, the swap between $s'$ and $s$ is disallowed. We iterate through these steps until no such pair of users can be exchanged. For brevity, we summarize the proposed joint optimization algorithm for the problem (\ref{obj}) in Algorithm 2. 
\begin{algorithm}
\small
\caption{Joint optimization algorithm for problem (\ref{obj})}
\begin{algorithmic}[1]
\REQUIRE An initial set $\cS$ and $\bar\cS = \cK \setminus \cS$.
\STATE Define $\cH=\{(s,s')|s\in\cS,~s'\in\bar\cS\}$.
\WHILE{$\cH\neq\emptyset$}
\STATE Update ${\cS}^{\dag} = \cS\cup\{s'\}\setminus\{s\}$;
\STATE Calculate the objective value under the scheduled user set ${\cS}^{\dag}$ in accordance with Algorithm 1;
\IF{${\cS}^{\dag}$ results in an improved system performance} 
\STATE Update $\cS={\cS}^{\dag}$.
\ENDIF
\STATE Update $\cH=\cH\setminus\{(s,s')\}$.
\ENDWHILE
\RETURN $\cS$.
\end{algorithmic}
\end{algorithm}

\section{Numerical Simulations}
This section presents numerical simulations that validate the efficiency of the proposed joint optimization strategy. 
Without loss of generality, the system parameters are configured as follows: we consider an MCN consisting of one BS and a varying number of users.
Unless otherwise specified, the number of users, number of subareas, number of subbands, and the weight in the objective function $\omega$ are given by $K = 20$,
$M = 10$, $N = 10$, and $\omega=0.5$, respectively.
Each subband has a bandwidth of $W_n=1\ \text{MHz}$, and the noise power spectral density is fixed at $-174\ \text{dBm/Hz}$. 
The adopted radio propagation model incorporates distance-based path loss, shadowing effects, and small-scale fading, as detailed in \cite{Greentouch}. Path loss adheres to the formula $128.1 + 37.6 \log_{10}(x)$, where $x$ represents the distance in meters between the BS and a randomly placed user, with distances drawn from a range of $50$ to $300$ meters. 
Shadowing follows a log-normal distribution with an $8\ \text{dB}$ standard deviation. Independent Rayleigh fading with a variance of $1$ characterizes the small-scale fading on each subband for each user.
In addition, User data rates for sensing tasks are randomly distributed between $10^5$ and $10^6$ bits per second, in accordance with \cite{rate}. 
The transmit power per user is uniformly distributed within the interval of $0.1$ to $0.2\ \text{Watts}$. The bit size of each sensing task is selected from the range $[10^3, 10^4]$.
The scaling factor of the normalization function $\text{Norm}(x,\eta)$ is $\eta = 10^6$.
The simulation results are obtained by averaging over $10^5$ samples of channel gains and transmit power of users. 

We refer to our optimal resource allocation algorithm, introduced in Section IV, for simplicity in performance comparisons, as the "Proposed" method. Additionally, we evaluate the following three benchmark approaches:
\begin{itemize} 
\item Benchmark 1: This is the method proposed in [5], where only the total latency of the system was considered.
\item Benchmark 2: In this approach, the $N$ users with the highest sensing rates are scheduled. 
$N$ subbands are randomly selected and assigned to the scheduled users.
A fractional task allocation method is employed, where the task allocation is given by $d_{k_n}=\bar d\frac{g_{k_n,n}}{\sum_{j\in\cN}g_{k_j,j}}$, $n\in\cN$.
\item Benchmark 3: In this benchmark, subbands are allocated to users with the highest channel gain for each subband.
The task allocation follows the fractional strategy as described in Benchmark 2.
\end{itemize}

\begin{figure}
	\centering
	\includegraphics[width=0.6\linewidth]{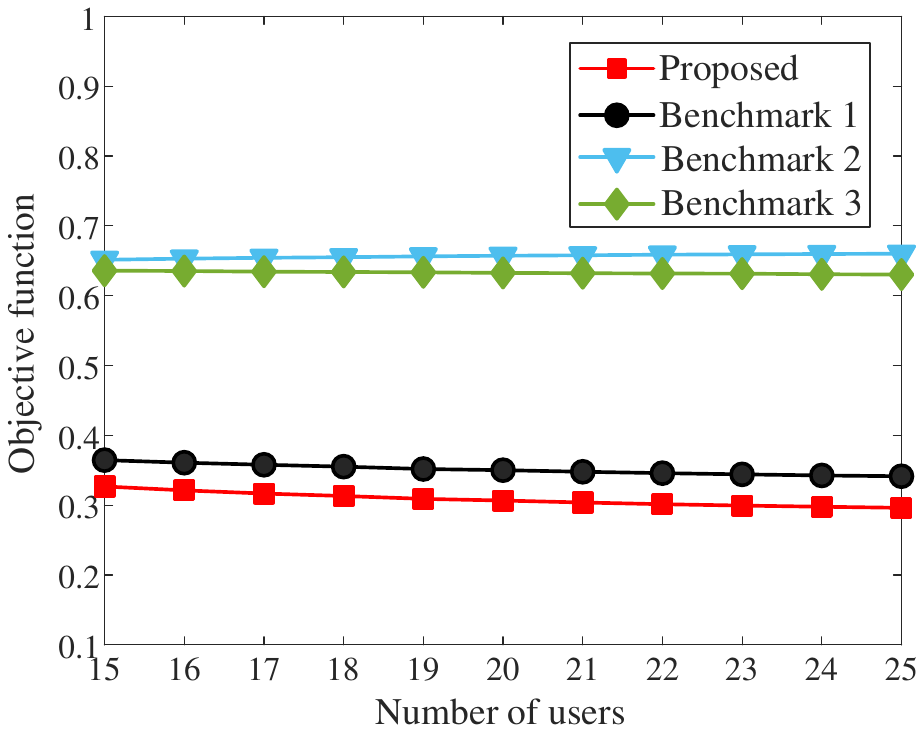}
	\caption{Objective function versus the number of users.}
	\label{user}
\end{figure}
From Fig. \ref{user} we can see that the objective function decreases as the number of users $K$ increases. 
This effect is more pronounced for the proposed method and Benchmark 1. 
The increased user diversity allows the system to select users with better sensing and transmission capabilities, thus minimizing overall latency. 
Additionally, the probability of covering more subareas increases, enhancing overall network performance.
The proposed method consistently outperforms the three benchmarks. 
Benchmark 2 and Benchmark 3, which rely on random and sensing-based allocations, perform worse due to the lack of consideration for latency and coverage.

\begin{figure}
	\centering
	\includegraphics[width=0.6\linewidth]{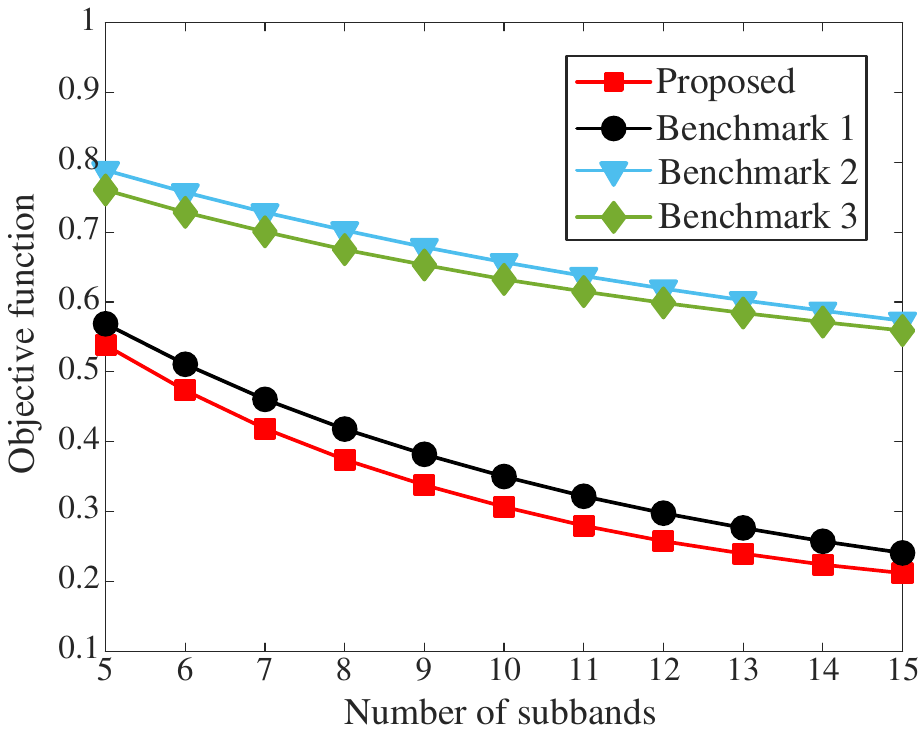}
	\caption{Objective function versus the number of subbands.}
	\label{subband}
\end{figure}
Fig. \ref{subband} shows that the objective function declines with an increasing number of subbands $N$. 
This is attributed to the diversity gain and the ability to schedule more users, thereby reducing total latency. 
The proposed method and Benchmark 1 show a steeper decline compared to Benchmark 2 and Benchmark 3.
The performance gain of the proposed method reaches $13.63\%$, $170.71\%$, and $164.07\%$ compared to Benchmark 1, Benchmark 2, and Benchmark 3, respectively.

\begin{figure}
	\centering
	\includegraphics[width=0.6\linewidth]{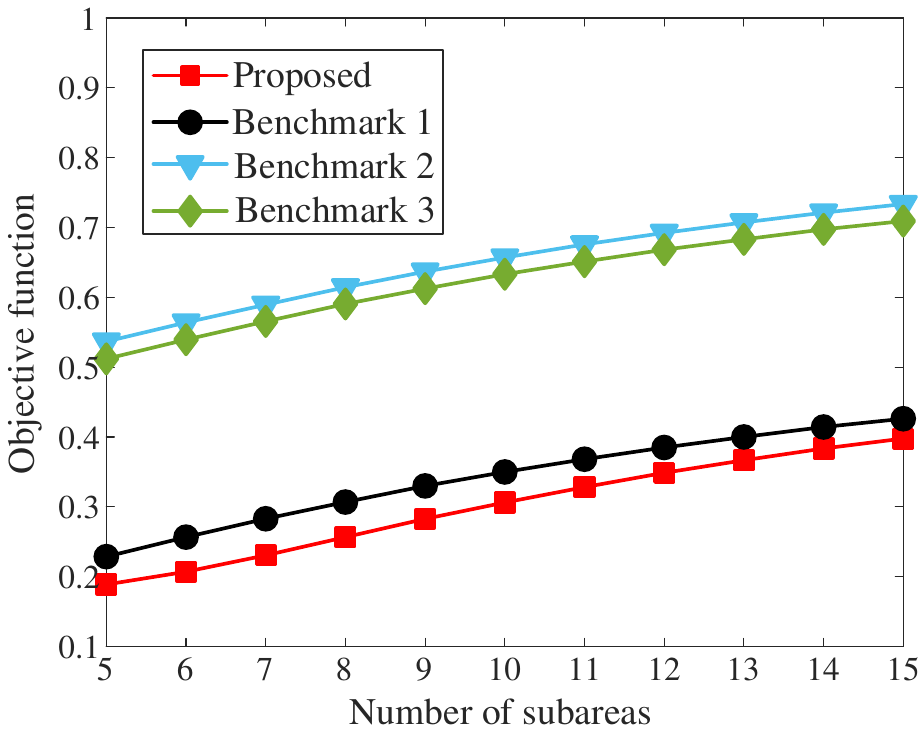}
	\caption{Objective function versus the number of subareas.}
	\label{subarea}
\end{figure}
Fig. \ref{subarea} demonstrates that the objective function increases as the number of subareas $M$ grows. 
A larger $M$ makes it more challenging to ensure full coverage by the scheduled users. Despite this, the proposed method consistently surpasses all benchmarks by effectively allocating subbands and distributing sensing tasks.
Benchmark 2 consistently underperforms because it prioritizes users with high sensing rates, neglecting the importance of transmission conditions and subarea coverage.

\begin{figure}
	\centering
	\includegraphics[width=0.6\linewidth]{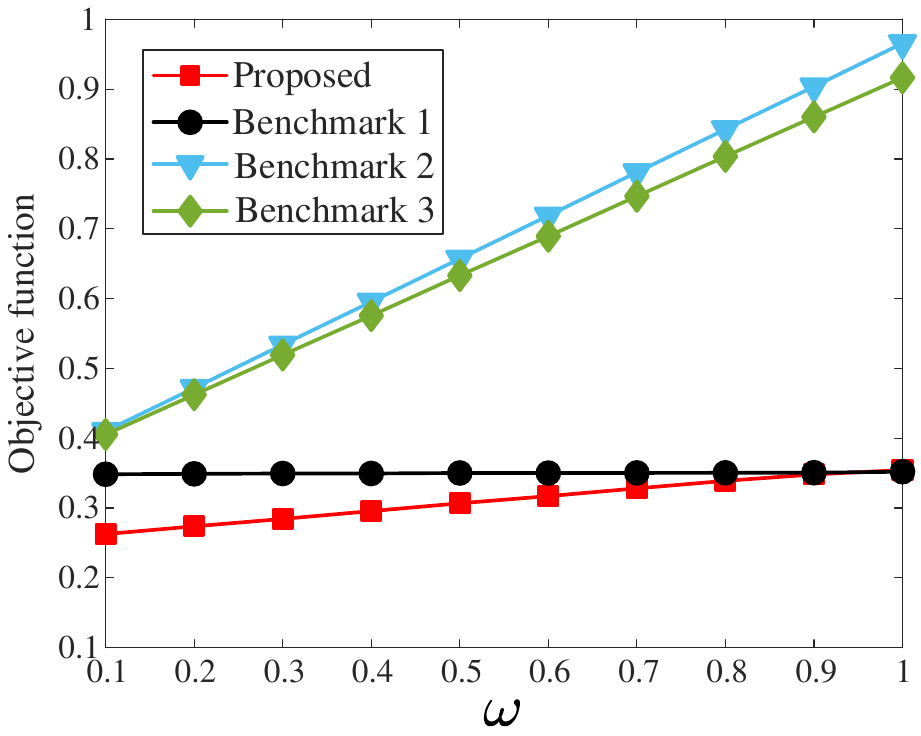}
	\caption{Objective function versus the weight $\omega$.}
	\label{omega}
\end{figure}
Fig. \ref{omega} highlights the trade-off between latency and coverage as controlled by the weight $\omega$.
When $\omega=1$, the focus is solely on minimizing latency, resulting in identical performance for the proposed method and Benchmark 1. 
However, as $\omega$ decreases, i.e., the emphasis on coverage increases, the proposed method demonstrates superior performance over Benchmark 1.
This highlights the key ability of the proposed scheme to optimize both the coverage and latency of the system. Moreover, since Benchmark 1 does not consider coverage, the objective value remains consistent with $\omega$.
\vspace{-0.2cm}
\section{Conclusion}
\vspace{-0.1cm}
This paper presented a novel coverage-aware resource allocation strategy for mobile crowdsensing networks, addressing the challenges of minimizing task latency and ensuring subarea coverage. 
By formulating the problem as a weighted minimization of latency and coverage gaps, we proposed an optimal subband and task allocation scheme, utilizing the Hungarian algorithm and a two-sided swapping method for efficient user scheduling. 
Simulation results demonstrated significant performance improvements over existing benchmarks, highlighting the effectiveness of our approach in enhancing the overall efficiency of mobile crowdsensing systems.


\end{document}